\documentclass[12pt,reqno]{article}

\usepackage[usenames]{color}
\usepackage{amssymb}
\usepackage{amsmath}
\usepackage{amsthm}
\usepackage{amsfonts}
\usepackage{amscd}
\usepackage{graphicx}

\usepackage[colorlinks=true,
linkcolor=webgreen,
filecolor=webbrown,
citecolor=webgreen]{hyperref}

\definecolor{webgreen}{rgb}{0,.5,0}
\definecolor{webbrown}{rgb}{.6,0,0}

\usepackage{color}
\usepackage{fullpage}
\usepackage{float}

\usepackage{graphics}
\usepackage{latexsym}
\usepackage{epsf}
\usepackage{breakurl}

\newcommand{\seqnum}[1]{\href{https://oeis.org/#1}{\rm \underline{#1}}}

\def\AND{\, \wedge \,}
\DeclareMathOperator{\real}{Re}
\DeclareMathOperator{\sgn}{sgn}
\def\Enn{\mathbb{N}}

\begin{document}

\theoremstyle{plain}
\newtheorem{theorem}{Theorem}
\newtheorem{corollary}[theorem]{Corollary}
\newtheorem{lemma}[theorem]{Lemma}
\newtheorem{proposition}[theorem]{Proposition}

\theoremstyle{definition}
\newtheorem{definition}[theorem]{Definition}
\newtheorem{example}[theorem]{Example}
\newtheorem{conjecture}[theorem]{Conjecture}

\theoremstyle{remark}
\newtheorem{remark}[theorem]{Remark}

\title{The Tribonacci constant and finite automata}

\author{Jeffrey Shallit\footnote{Research supported by NSERC grant 2024-03725.}\\
School of Computer Science\\
University of Waterloo\\
Waterloo, ON N2L 3G1 \\
Canada\\
\href{mailto:shallit@uwaterloo.ca}{\tt shallit@uwaterloo.ca}}

\maketitle

\vskip .2 in
\begin{abstract}
We show that there is no automaton accepting the Tribonacci representations of $n$ and $x$ in parallel, where $\psi = 1.839\cdots$ is the Tribonacci constant, and $x= \lfloor n \psi \rfloor$.  Similarly, there is no Tribonacci automaton
generating the Sturmian characteristic word with slope $\psi-1$.
\end{abstract}

\allowdisplaybreaks

\section{Introduction}

The Fibonacci numbers are defined by the initial values
$F_0 = 0$, $F_1 = 1$, and the recurrence $F_i = F_{i-1} + F_{i-2}$
for $i \geq 2$.
We can define a numeration system---sometimes called the Zeckendorf
numeration system--- based on $F_i$ for $i \geq 2$, writing every integer $n$ uniquely as a sum
of distinct Fibonacci numbers, subject to the rule that we never
use both $F_i $ and $F_{i+1}$.  See \cite{Lekkerkerker:1952,Zeckendorf:1972}.  

Let $\varphi = (1+\sqrt{5})/2$, the golden ratio, be the positive
zero of the polynomial $X^2-X-1$.  It is well-known that the
sequence $(\lfloor \varphi n \rfloor)_{n \geq 0}$ is
Fibonacci-synchronized, in the sense that there is a finite
automaton that takes as input, in parallel, the Zeckendorf representations
of $n$ and $x$ and accepts if and only if $x = \lfloor \varphi n \rfloor$.
(In fact, this is true for any quadratic irrational; see
\cite{Schaeffer&Shallit&Zorcic:2024}.)
Furthermore, the Sturmian sequence
$(\lfloor (\varphi-1)(n+2) \rfloor - \lfloor (\varphi-1) (n+1) \rfloor)_{n \geq 0}$
is Fibonacci-automatic; there is an automaton that takes $n$ as input
in Zeckendorf representation and computes the $n$'th term of the
Sturmian sequence.\footnote{Usually Sturmian sequences are indexed
starting at $1$; we have modified the definition slightly to start
the indexing at $0$.}

It is then natural to ask whether the same kind of thing holds for generalizations
of the Fibonacci numbers; for example, the Tribonacci numbers.
The answer is ``no'' (Section~\ref{main}), but also ``sort-of yes''
(Section~\ref{end}), depending on 
what generalization you demand.

\section{The Tribonacci numbers}

The Tribonacci numbers are defined by the initial values
$T_0 = 0$, $T_1 = 1$, and $T_2 = 1$, and the recurrence
$T_i = T_{i-1} + T_{i-2} + T_{i-3}$ for $i \geq 3$.  (Other authors use
different indexing.)  They have a well-known closed form, as follows:
$$ T_n = c_1 \psi^n + c_2 \alpha^n + c_3 \beta^n,$$
where
where $\psi, \alpha, \beta$ are the zeros of the polynomial
$X^3-X^2-X-1$, with $\psi = 1.839 \cdots$ the unique real zero
and complex conjugates $\alpha = -0.41964\cdots - 0.60629\cdots i$
and $\beta = -0.41964 \cdots + 0.60629\cdots i$ the two complex
zeros.  See, for example, \cite{Spickerman:1982}.

Here 
\begin{align*}
c_1 &= {{\psi} \over {(\psi-\alpha)(\psi-\beta)}} =
0.336228 \cdots \\
c_2 &= {{\alpha} \over {(\alpha-\psi)(\alpha-\beta)}}  =
-0.16811 \cdots + 0.19832 \cdots i \\
c_3 &= {{\beta} \over {(\beta-\psi)(\beta-\alpha)}}  = 
-0.16811 \cdots - 0.19832 \cdots i \\
\end{align*}
Notice that $|\alpha| = | \beta | = 0.7373527\cdots < 1$, and hence
$T_n = c_1 \psi^n + o(1)$.

\section{Tribonacci automata}

In analogy with Zeckendorf representation,
natural numbers can be represented uniquely as a sum of 
distinct Tribonacci numbers $\sum_{2 \leq i \leq t} e_i T_i$
where $e_i \in \{0,1\}$, provided we never use three consecutive
Tribonacci numbers in the representation.  Such a representation
can be denoted as a binary word of the form $e_t e_{t-1} \cdots e_2$.
See \cite{Carlitz&Scoville&Hoggatt:1972}.

We will need two sorts of automata that take the Tribonacci representation
of natural numbers as inputs.

Let $f: \Enn \rightarrow \Enn$ be a function.
The first kind of automaton we need
is a synchronized DFA (deterministic finite automaton)
that takes $n$ and $x$ in parallel
as input (represented in Tribonacci form) and accepts if and only
if $x = f(n)$.  In this case we think of the automaton as computing
the function $f(n)$.

The second kind of automaton we need is a DFAO (deterministic finite automaton
with output).  This is like a DFA, but has outputs associated with states.
In this case we can think of the DFAO as computing a sequence $(a_n)_{n \geq 0}$ over a finite alphabet:  in comes the Tribonacci representation of $n$,
and $a_n$ is the output associated with the last state reached.

\section{The main theorem}
\label{main}

Define $a(n) = \lfloor \psi n \rfloor$ and
$b(n) = \lfloor (\psi-1) (n+2) \rfloor - \lfloor (\psi-1) (n+1)  \rfloor$
for $n \geq 0$.
Our main result
is the following:  
\begin{theorem}
\leavevmode
\begin{itemize}
\item[(a)] The sequence $(a(n))_{n \geq 0}$ is not Tribonacci-synchronized.
\item[(b)] The sequence $(b(n))_{n \geq 0}$ is not Tribonacci-automatic.
\end{itemize}
\end{theorem}

\begin{proof}
Suppose $(a(n))_{n \geq 0}$ is Tribonacci-synchronized.  This means
there is an automaton $A$ accepting exactly those pairs $(n,x)$
such that $x = \lfloor \psi n \rfloor$.  Then we could construct
construct an automaton for $(b(n))_{n\geq 0}$ by translating the
first-order formula
$$ \exists y, z \ A(n+2,y) \AND A(n+1,z) \AND z=y+2 $$
into an automaton using B\"uchi's theorem
\cite{Buchi:1960,Bruyere&Hansel&Michaux&Villemaire:1994}.

To prove (a) it
suffices to show that $(b(n))_{n \geq 0}$ is not Tribonacci-automatic.
If it were, then $c(n) = \lfloor \psi (n+1) \rfloor - \lfloor \psi n  \rfloor$
would also be Tribonacci-automatic, as $(c(n))_{n \geq 0}$ is just a shift
and recoding of the values of $(b(n))_{n \geq 0}$.
So it suffices to show $(c(n))_{n \geq 0}$ is not Tribonacci-automatic.
We suppose it is and get a contradiction.  This will prove both (a) and (b).

Let $C$ be a Tribonacci automaton accepting $c(n)$.  Consider the states $q_i$
of $C$ reached on inputs of the form $1 0^i$ for $i \geq 0$; these 
correspond to inputs that are Tribonacci numbers.  We will show that
all of these states are pairwise distinguishable in the Myhill-Nerode
sense \cite[\S 3.4]{Hopcroft&Ullman:1979}
by a word of the form $0^k$.  More precisely, we will show that
for all $i$ and $j$, if $i<j$ then there exists $k$ such that
$c(T_{i+k})$ and $c(T_{j+k})$ take different values.

Since $1 < \psi < 2$, it is clear that $c(i) \in \{1,2\}$.  It is also easy
to see that $c(i) = 1$ if and only if $\{ \psi i \} \in [0, 2-\psi)$.  
On the other hand, from the explicit formula for Tribonacci numbers
we have
\begin{align*}
\psi T_n - T_{n+1} &=  
\psi(c_1 \psi^n + c_2 \alpha^n + c_3 \beta^n) - (
c_1 \psi^{n+1} + c_2 \alpha^{n+1} + c_3 \beta^{n+1} ) \\
&= c_2(\psi-\alpha)\alpha^n + c_3(\psi-\beta)\beta^n ,
\end{align*}
and hence
\begin{align*}
\{ \psi T_n \} &= \{ c_2(\psi-\alpha)\alpha^n + c_3(\psi-\beta)\beta^n \} \\
&= \{ 2 \real c_2 (\psi-\alpha) \alpha^n \}.
\end{align*}
Now $| c_2 (\psi-\alpha) | = 0.608085\cdots$ and
$|\alpha| = 0.73735 \cdots$, so it is easy to check that
$|2 \real c_2 (\psi-\alpha) \alpha^n| \rightarrow 0$ as
$n \rightarrow \infty$, and furthermore
$ |2 \real c_2 (\psi-\alpha) \alpha^n| < 2-\psi$ for $n \geq 5$.
Thus we see that $c(T_n) = 1$ for $n \geq 5$
if and only if $\sgn \real c_2 (\psi-\alpha) \alpha^n = +1$.

Now let $\gamma = \arg c_2 (\psi-\alpha) = 2.536155\cdots$
and $\zeta = \arg \alpha = 4.10695 \cdots$.  It then
follows that $c(T_n)$, for $n \geq 5$, and hence
$\sgn \real c_2 (\psi-\alpha) \alpha^n$,
depends on the value of $v(n) := \gamma + n \zeta \bmod 2\pi$.
If $0 < v(n) < \pi/2$, or $v(n) > 3\pi/2$, then
$c(T_n) = +1$; otherwise $c(T_n) = 2$.
(Note that $\zeta$ and $2 \pi$ are linearly independent over the rationals;
if there were a rational relation connecting them, then
$\alpha/|\alpha|$ would be a root of unity, which it is not; it satisfies
the equation $X^{12} + 4X^{10} + 11X^8+ 12X^6 + 11X^4 + 4X^2 + 1$.)

We now claim that for all $d \geq 1$ there exist infinitely many $m$
such that $c(T_m) \not= c(T_{m+d})$.   There are two cases.

\bigskip

\noindent{\it Case 1:}  $0 < \zeta d \bmod 2\pi < \pi$.
Let $\tau = \pi-(\zeta d \bmod 2\pi) > 0$.
By Kronecker's theorem 
\cite[Thm.~438]{Hardy&Wright:1985}, there
are infinitely many $m$ such that $m \zeta \bmod 2\pi \in (\pi/2 - \tau, \pi/2)$.
For such $m$ we have $c(T_m) = 1$ and $c(T_{m+d}) = 2$.

\bigskip

\noindent{\it Case 2:}  $\pi < \zeta d \bmod 2\pi < 2\pi$.
Let $\tau = 2\pi - (\zeta d \bmod 2\pi) > 0$.
By Kronecker's theorem there are
infinitely many $m$ such that $m \zeta \bmod 2\pi \in (3\pi/2-\tau, 3\pi/2)$.
For such $m$ we have $c(T_m) = 2$ and $c(T_{m+d}) = 1$.

\bigskip

Now let $i < j$ and set $d = j-i$.  By the above we can find infinitely
many $m$ such that $c(T_m) \not= c(T_{m+d})$.  Choose $m > i+5$ and
set $k = m-i$.   Then $c(T_{i+k}) = c(T_m) \not= c(T_{m+d}) = c(T_{j+k})$,
as desired.

Hence the states $q_i$ and $q_j$ of the automaton $C$
are distinguished by the word
$0^k$, and then by the Myhill-Nerode theorem there is no finite
automaton computing $c(n)$.

\end{proof}

\section{Implications for logic}

Mousavi et al.~\cite{Mousavi&Schaeffer&Shallit:2016} and, independently,
Hieronymi \cite{Hieronymi:2016}, showed
that the map $n \rightarrow \lfloor \varphi n \rfloor$ for
$\varphi = (1+\sqrt{5})/2$ is expressible in the first-order theory
of the structure $\langle \Enn, +, V(n) \rangle$, where $V(n)$ denotes the
smallest Fibonacci number $F_i$ in the Zeckendorf representation of $n$.
As a consequence the first-order theory
of $\langle \Enn, +, n \rightarrow \lfloor \varphi n \rfloor \rangle$ is decidable.
Later this was also shown by Khani and Zarei \cite{Khani&Zarei:2023} by a different method.

It would be natural to suspect that, analogously, the
map  $n \rightarrow \lfloor \psi n \rfloor$ would be expressible in
the first-order theory
of the structure $\langle \Enn, +, V'(n) \rangle$ where $\psi$ is as in previous sections, and $V'(n)$ is the smallest Tribonacci number 
$T_i$ appearing in the Tribonacci
representation of $n$.  However, we have the following:
\begin{theorem}
The map $n \rightarrow \lfloor \psi n \rfloor$ is not expressible
in the first-order theory of $\langle \Enn, +, V'(n) \rangle$.
\end{theorem}

\begin{proof}
If $n \rightarrow \lfloor \psi n \rfloor$ were expressible in the
first-order theory of 
$\langle \Enn, +, V'(n) \rangle$, then essentially the same proof as in
\cite{Bruyere&Hansel&Michaux&Villemaire:1994}, adapted for Tribonacci
representation instead of base $k$, shows that there would be a computable
synchronized automaton taking $n$ and $x$ in parallel as inputs,
and accepting if $x = \lfloor \psi n \rfloor$.  But we have shown
there is no such automaton.
\end{proof}

\section{Final words}
\label{end}

The sequence \seqnum{A352719} and \seqnum{A352748} in the
On-Line Encylopedia of Integer Sequences (OEIS) \cite{oeis} record the
integers $n$ such that (in our notation) $T_n - \psi T_{n-1} > 0$
(resp., $T_n - \psi T_{n-1} < 0$).

Ultimately, what makes the Fibonacci and Tribonacci cases quite
different is that the sign of $F_{n+1} - \varphi F_n$ is periodic,
but the sign of $T_{n+1} - \varphi T_n$ is not periodic.  We expect that
similar behavior will occur in some other types of numeration
systems; for example, those based on cubic Pisot numbers having
two complex conjugates.

Although there is no synchronized Tribonacci automaton for
the sequence $n \rightarrow \lfloor \psi n \rfloor$, 
there {\it is\/} an automaton for ``almost'' this sequence.  Namely, as
shown in \cite[\S 10.12]{Shallit:2023}, there is an automaton
of $10$ states that accepts, in parallel, $n$ and the position $A_n$
of the $n$'th occurrence of the symbol $0$ 
in the infinite Tribonacci word ${\bf TR} = 0102010\cdots$ (where indexing of $\bf TR$ is done
starting at $1$).  Furthermore it is known 
\cite{Dekking&Shallit&Sloane:2020}
that
$$ \lfloor \psi n \rfloor - 1 \leq A_n \leq \lfloor \psi n \rfloor + 1.$$
Thus if we are happy with such an approximation to
$n \rightarrow \lfloor \psi n \rfloor$, an automaton suffices.

This may explain, in 
part, why Turner \cite{Turner:1989}
was able to find a clean generalization of the
Wythoff pairs to the Tribonacci case when using a function like
$A_n$, but could not succeed when using $\lfloor \psi n \rfloor$.

\section*{Acknowledgments}

I thank Joris Nieuwveld, Jason Bell, and Luke Schaeffer for their comments.

\end{document}